\documentclass[a4paper,UKenglish]{article}
\usepackage[utf8]{inputenc}
\usepackage[T1]{fontenc}
\usepackage{authblk}
\usepackage[margin=1.5in]{geometry}

\title{Finite-Degree Predicates and Two-Variable First-Order Logic}
\author[1]{Charles Paperman}
\affil[1]{
University of Warsaw}
\usepackage{stmaryrd}
\usepackage{mathdots}
\usepackage{amssymb}
\usepackage{amsthm}
\usepackage{amsmath}
\usepackage{amsfonts}
\usepackage{lipsum}
\usepackage{mathtools}

\DeclarePairedDelimiter\floor{\lfloor}{\rfloor}
\theoremstyle{plain}
\newtheorem{theorem}{Theorem}
\newtheorem{proposition}[theorem]{Proposition}
\newtheorem{corollary}[theorem]{Corollary}
\newtheorem{lemma}[theorem]{Lemma}

\newtheorem*{lemma*}{Lemma}

\usepackage{float}
\usepackage{multirow}
\usepackage{here}
\usepackage{ae,aeguill}
\usepackage{fancyhdr}
\usepackage{wrapfig}
\usepackage{colordvi}
\usepackage{graphicx}
\usepackage{fancybox}
\usepackage{moreverb}
\usepackage{pgf}
\usepackage{tabularx}
\usepackage{tikz}
\usetikzlibrary{calc,automata,arrows,chains,matrix,positioning,scopes,decorations.pathreplacing, decorations.pathmorphing}

\makeatletter

\newcommand{\N}{{\mathbb N}}
\renewcommand{\leq}{\leqslant}
\renewcommand{\geq}{\geqslant}

\newcommand{\cP}{\mathcal{P}}

\newcommand{\cF}{\mathcal{F}}
\newcommand{\arb}{\mathrm{Arb}}
\newcommand{\bitadd}{\mathrm{MSB}}

\newcommand{\FO}{\mathbf{FO}}

\newcommand{\EF}{Ehrenfeucht-Fraïssé }

   {%
   \begin{list}{--}%
   {\setlength{\leftmargin}{0cm}%
   }%
   #1%
   }%
   {%
   \end{list}%
   }%

\newenvironment{conditions}
{%
	\begin{list}{\rm (\theenumi)}%
	{\noindent%
		\usecounter{enumi}%
		\setlength{\topsep}{2pt}%
		\setlength{\partopsep}{0pt}%
															 \setlength{\itemsep}{2pt}%
		\setlength{\parsep}{0pt}%
		\setlength{\leftmargin}{2.5em}%
		\setlength{\labelwidth}{1.5em}%
		\setlength{\labelsep}{0.5em}%
		\setlength{\listparindent}{0pt}%
		\setlength{\itemindent}{0pt}%
	}%
}%
{\end{list}}%

\newcommand{\pneg}[1]{#1_-}
\newcommand{\ppos}[1]{#1_+}
\DeclareMathOperator{\ttau}{\mkern-2 mu-\mkern-1 mu \tau}

\newcommand{\St}{\cP_3}

\DeclareFontFamily{U}{mathx}{\hyphenchar\font45}
\DeclareFontShape{U}{mathx}{m}{n}{
      <5> <6> <7> <8> <9> <10>
      <10.95> <12> <14.4> <17.28> <20.74> <24.88>
      mathx10
      }{}
\DeclareSymbolFont{mathx}{U}{mathx}{m}{n}
\DeclareMathSymbol{\bigtimes}{1}{mathx}{"91}

\newcommand{\Cl}[1]{\textrm{Cl}(#1)}

\begin{document}
\maketitle
\begin{abstract}
	We consider two-variable first-order logic on finite words with a fixed number of quantifier alternations. We show that all languages with a {neutral letter} 	definable using the order and {finite-degree} predicates are also definable  with the order predicate only. From this result we 	derive the separation of the alternation hierarchy of two-variable logic on this signature.
\end{abstract}
\section{Introduction}

	Finite model theory and the lower classes of circuit complexity are intricately interwoven.
	In the context of circuit complexity, logics are considered over finite words with \emph{arbitrary numerical predicates}.
	Intuitively, we allow the use of any predicate that only depends on the size of the word.
	A first result
	from Immerman~\cite{Immerman87} provides an equivalence between languages definable by
	first-order logic enriched with
	arbitrary numerical predicates  on the one hand,
	and languages computable by families of circuits of constant depth and polynomial size
	on the other. Since then, several  meaningful
	 circuit complexity classes have been shown to be equivalent to
	 logical fragments~\cite{Bar92,KLPT06}.  It is therefore possible
	to obtain deep and interesting inexpressibility results by using circuits lower bounds.

	For instance, by using a famous lower bound for the \emph{parity} language~\cite{FSS84}, Barrington, Compton,
	Straubing and Thérien~\cite{Bar92} showed that the regular languages definable in first-order logic with arbitrary numerical predicates
	are definable with only the \emph{regular predicates}. Relying on an algebraic
	description of first-order logic with regular predicates, it is possible to decide the definability
	of a regular language in this logic.

	Conversely, it is tempting
	to use finite model theory methods to compute circuit lower bounds.  This approach has achieved
	relative success for
	\emph{uniform} versions of circuit complexity classes. For instance,
	 Roy and Straubing
	provide a separation result for the long-standing question of the separation of
	$\mathbf{ACC}$ from $\mathbf{NC^1}$ in a
	{highly uniform setting}~\cite{RS06}.
	In these settings, this uniformity condition has two different interpretations:
	\begin{enumerate}
		\item In the circuit framework, it is a restriction on the complexity of
		the wiring of the gates.
		\item In the logical framework, it is a restriction on the class of numerical
		predicates considered in the fragment.
	\end{enumerate}
	In order to deal with the combinatorics of arbitrary numerical predicates,
	the languages with a \emph{neutral letter}
	have been introduced in~\cite{BILST05}. Formally, a language $L$ has a neutral letter~$c$ if for any pair of words $u,v$, we have~$ucv\in L$ if, and only if, $uv \in L$. Less formally, this letter~$c$ can be added or removed anywhere in a word without changing its membership to~$L$.
	The underlying idea was that numerical predicates would be essentially useless in the presence of a neutral letter.
	This was made formal through the \emph{Crane Beach conjecture}:

	\emph{Every language with a neutral letter definable in first-order logic with arbitrary numerical predicates is definable in first-order
	logic with the linear order only}.

	Furthermore, some of the most interesting
	languages, such as the \emph{parity} language, possess a neutral letter.
	Unfortunately, this conjecture has been disproved in the article~\cite{BILST05} in the
	context of first-order logic, as long 	as the Bit predicate is in the signature. This result
	prevents the use of this approach to obtain circuit lower bounds for more expressive
	classes. However, for fragments of first-order logic the Crane Beach conjecture
	is still of interest. For instance, the Crane Beach conjecture holds for
	the fragment without quantifier
	alternation~\cite{BILST05}.

	Turning to other fragments, two-variable first-order logic is a robust and well-studied class
	that offers a wide range of  long-standing and intriguing open questions. It is not know
	whether the Crane Beach conjecture holds for this fragment. This question is related
	to a long-standing open linear lower bound for the addition function, since two-variable logic
	is equivalent to linear circuits of $\mathbf{AC}^0$~\cite{KLPT06}. Therefore,
	if the Crane Beach conjecture holds, then the addition function is not
	computable by a constant-depth linear-size circuit family.
	This result would improve
	on a known lower bound for  addition that states that addition is not computable by circuits
	of constant depth with a linear number of wires~\cite{CFL83}.
	We remark that lower bound for addition has been discussed and informally mentioned several
	times~\cite{RW91,CR96,KPT05,KLPT06} and formally
	stated in the article~\cite[Open problem 23]{Koucky09}.

	In this paper, we focus on the case of two-variable logic, which is poorly understood in this context.
	We first prove that languages with a neutral
	letter definable in two-variable logic with arbitrary numerical predicates
	can be defined allowing only the linear order and the following predicates:
	\begin{enumerate}
		\item The class $\cF$ of finite-degree predicates,  that is, binary predicates that are relations
	over integers and such that each vertex of their underlying infinite
	directed graph has a finite degree.
		\item The predicate $\bitadd_0$ defined as follows.
		The predicate $\bitadd_0$ is true of $x$ and $y$
		if the binary representation of $y$ is obtained by zeroing the most significant bit of $x$.
		 More formally
			$$\bitadd_0 =\{ (x,x-2^i)\mid x\in \N, \text{ and }i=\floor{\log(x)}\text \}\enspace.$$
	\end{enumerate}
	As an intermediate step toward a better comprehension of the Crane Beach conjecture for $\FO^2$,
	we propose to study the relationship between $<$ and $\cF$, and present a Crane Beach result
	which is thus one predicate shy from showing the Crane Beach conjecture for $\FO^2$ over arbitrary
	numerical predicates.

	The main result of this paper is a proof of
	the Crane Beach conjecture for each layer of the alternation hierarchy of the
	two-variable first-order logic equipped with the linear order and
	the {finite-degree predicates}.

	Note that the general arbitrary numerical predicates in the statement would entail a long standing conjecture  on the circuit complexity of the addition function. Thus, this result can be viewed as a uniform version of this circuit lower bound.
	This result immediately implies that this hierarchy is strict.
	This provides, to the best of our knowledge, the first example of a Crane Beach conjecture
	that applies to each level of an alternation hierarchy.  Ramsey's
	Theorem for $3$-hypergraphs will be our key combinatorics tool.
	This theorem indicates that the Crane Beach conjecture for $\FO^2$ hinges on the interaction
	between finite-degree predicates and the predicate $\bitadd_0$.  \\

	\noindent\textbf{On the two-variable restriction}:\\
	It is already known that the first-order
	logic with the ``$+$'' predicate satisfies the Crane Beach conjecture. Furthermore,
	the  $\bitadd_0$ predicate is definable in first-order logic with the predicate ``$+$''
	and the unary predicate $\{2^x\mid x\in \N\}$. The proof of the Crane Beach conjecture
	for ``$+$'' predicate can be augmented to handle this extra unary numerical predicate.
    Therefore,
	 we deduce that the first-order logic with the order and the $\bitadd_0$ predicate also satisfies
	the Crane Beach conjecture.

	The case of finite-degree predicates is more intricate. Indeed, even if this class of
	predicates satisfies a form of locality, it is still not known if the Crane Beach conjecture hold for
	$\FO[<,\cF]$. This class contains numerous expressive numerical
	predicates as the \emph{translated bit predicate} which is true in positions $(x,y)$
	if the $(y-x)^\text{th}$ bit of $x$ is a one. The Crane Beach conjecture
	may holds for finite-degree predicates but the classical proof, e.g. \emph{collapse on active domain},
	seems to fail~\cite{BILST05,RS06,KS12b}.  \\

	\noindent\textbf{Organization of the paper:}\\
	 Section~\ref{section:def} is dedicated to the necessary definitions. In Section~\ref{section:first_res} we present an \EF game adapted to our
	context.  We present
	in Section~\ref{section:main} our main result with immediate corollaries. The final section
	is dedicated to the proof.

\section{Definitions}\label{section:def}
  	 A finite word $u=u_0\cdots u_{n-1}$ of $A^*$ is represented by a relational
  	 structure on the set  $\{0,\cdots,n-1\}$ over the vocabulary consisting of the \emph{letter predicates}
  	 $\{\textbf{a}\mid a\in A\}$
  	 and of the \emph{numerical predicates}.
  	 On the one hand,
  	 the letter predicate $\textbf{a}$ is interpreted as the subset of all the positions labelled by the
  	 letter $a$. On the other hand,
  	a {numerical} predicate
	interpretation only depends on the size $n$ of the input word.
	Therefore, an interpretation of the predicate
	symbol $\mathbf{P}$ of arity $k$ is a sequence $P = (P_n)_{n}$,
	where $P_n \subseteq \{0,\ldots,n-1\}^k$. 	Note that $\mathbf{P}$ is a syntactic object,
	while $P$ is its interpretation.
	Furthermore a numerical predicate is said to be \emph{uniform} if it can be seen
	as a relation on integers. More precisely, a numerical predicate $P=(P_n)_n$ of arity $k$
	is uniform if there exists an integer relation $Q\subseteq \mathbb{N}^k$
	satisfying $Q\cap \{0,\ldots,n-1\}^k = P_n$. From now on, we do
	not distinguish numerical predicates from their interpretation and uniform
	predicates are seen as relations on integers. The class of all numerical predicates
	is denoted by $\arb$. Remark that the word \emph{uniform} in this context is not related
	to the classical notion of \emph{uniformity} in circuit complexity. \\

	\noindent\textbf{Examples:}
	\begin{itemize}
		\item	The classical predicates $x<y$ or $x+y=z$ and $xy=z$ are numerical predicates and are uniforms.
		\item The predicate $x+y=\max$, where $\max$ is the last position of the word, is not uniform.
	\end{itemize}

  	 The logical formulae we consider are the first-order formulae over finite words.
  	 They are obtained with the following grammar:
  	 \begin{displaymath}
  	 \varphi= \mathbf{a}(x)\mid \mathbf{P}(x_1,\ldots,x_k)\mid \varphi\land \varphi \mid \lnot \varphi\mid
  	 \exists x\ \varphi\enspace.
  	 \end{displaymath}
  	 Here $x,x_1,x_2,x_3,\ldots$ denote first-order variables, which are interpreted by positions in the word.
	The letter predicate $\mathbf a(x)$, is interpreted by
	``the letter in position~$x$ is an a,''
	and $\mathbf{P}(x_1,\ldots,x_k)$, is interpreted by ``the  predicate $P$ is true on $(x_1,\ldots,x_k)$.''
	As usual, the Boolean connectives $\wedge$ and $\neg$ are interpreted by ``and'' and ``not,'' respectively,
	and $\exists x$ as a first-order existential quantification.
	We use the standard notation $u\models \varphi$ to signify that the word
	$u$ satisfies the formula $\varphi$. We also denote by $u\models \varphi(i)$
	if the formula $\varphi(x)$ is \emph{true} when its  free variable is interpreted by the integer
	$i<|u|$. The \emph{quantifier depth} of a formula is the maximal number of nested quantifiers.

	Let $\cP$ be a class of numerical predicates. We denote by $\FO[\cP]$ the class of first-order formulae
	that use numerical predicates in $\cP$. We also denote by $\FO^2[\cP]$
	the subclass of formulae of $\FO[\cP]$ that use only two variables but allows the reuse of them.
	We say that a language $L$ is definable in a fragment of logic if there exists a formula
	in this fragment such that $L$ is the language of words satisfying this formula.\\

	\noindent\textbf{Example:}\\
	The language $A^*aA^*bA^*cA^*$ can be described by the first-order formula
	$$\exists x\ \exists y\ \exists z\ x <y < z \land \mathbf{a}(x)\land \mathbf{b}(y)\land \mathbf{c}(z)\enspace.$$
	This formula uses three variables $x,y$ and $z$. However, by reusing $x$ we can rearrange it so that it uses two variables:
	\begin{align}
	\exists x\ \mathbf{a}(x) \land \Big(\exists y\  x < y  \land \mathbf{b}(y)\land \big(\exists x \ y<x\land \mathbf{c}(x)\big)\Big)\label{eq1}
	\end{align}

	The alternation hierarchy of $\FO^2$ is also of interest here.
	To define formally the number of alternations of a formula,
	it is not possible to use prenex canonical normal form obtained  by applying DeMorgan’s laws
	to move negations past conjunctions,
	disjunctions and quantifiers. Indeed, these constructions increase the number of variables. That said, the number
	of alternations is still a relevant parameter that could be defined as follows:
	Consider the tree naturally associated to a formula, as the grammar previously exposed. For instance,
	formula \eqref{eq1} has ``$\exists$'' as a root  and the atomic formulae as the leaf.
	In a two-variable first-order formula we count the maximal number of alternations between the
	root and the leaves
	once the negations have been pushed on to the leaves. A more precise definition
	could be found in the article~\cite{IW09}. We denote by $\FO^2_k[\cP]$
	the formulae of $\FO^2[\cP]$ that have at most $k$ quantifier alternations.
	The hierarchy induced by $\FO^2_k[<]$ is known to be strict~\cite{IW09} and its membership
	problems is decidable~\cite{KS12,KW12}. Without loss of generality,
	we will always consider two-variable logic over
	predicates of arity at most $2$. \\

\section{\EF game}\label{section:first_res}
	One of the  important tools for proving our main result is the \EF game for two-variable logic.
	It is often used in the context of finite model theory to  show certain inexpressibility
	results. Libkin's book~\cite{Libkin_book} provides a good exposition.
	In this section, we present the \EF  game and  briefly sketch a proof that the Crane Beach conjecture
	holds for $\FO^2_m[<,+1]$. This could be easily proved
	by using some algebraic descriptions of $\FO^2_m[<,+1]$ obtained by Kufleitner and Lauser~\cite{KL13}
	but we prove it using \EF game as an introduction to our general result.

	In the context of two-variable logic with a bounded number of alternations $m$ and quantifier depth $s$,
	the associated \EF game is defined as follows:
	\begin{itemize}
		\item The game is played by two players: \emph{Spoiler} and \emph{Duplicator}, on two relational structures.
		In our case, the relational structures are associated with the words $u$ and $v$
		equipped with the letter predicates
		and a finite  number of numerical predicates.
		\item The first round starts with Spoiler,  who chooses either $u$ or $v$ and plays by putting a pebble
		on a position. Then Duplicator chooses the other word and puts a pebble on one of its positions.
		\item The subsequent rounds proceed as follows: each word is labelled by at most two pebbles.
		First, the two oldest pebbles are removed. Then,
		 Spoiler plays on one structure and Duplicator on the other. If the relational
		structures induced by the two pairs of pebbles are not isomorphic, Spoiler wins.
		\item During all the game, Spoiler can change at most $m$ times between
		the two words. Duplicator wins the game
		if he did not loose the game before the end of the $s^\text{th}$ round.
	\end{itemize}
	We say that Spoiler has a \emph{winning
	strategy} if he has a strategy that allows him to win the game whatever
	Duplicator plays.
	The following theorem 	is a well-known result
	that could be easily adapted, for instance,
	from the book~\cite{Libkin_book}.
	\begin{theorem}~\label{theo:ef}
		A language $L$ belongs to $\FO^2_m[\cP]$ if and only if there exist predicates
			$P^1,\ldots,P^t\in\cP$ and $s\in \N$ so that for any words $(u,v)\in L\times L^\text{c}$
			Spoiler has a winning strategy for the two-pebble game with $s$ rounds and
			 $m$ alternations on $(u,v)$ over the predicates $P^1,\ldots,P^t$.
	\end{theorem}

	This theorem is our main interface to logic in order to establish Crane Beach-like results.
	The proof method we are going to sketch is a  rather classical \emph{back-and-forth} construction.
	As we mention before, the next result is also a direct consequence of known algebraic
	characterisations of these fragments~\cite{KL13}.
	\begin{proposition} \label{prop:succ}
		For any $m$, languages with a neutral letter in $\FO_{m}^2[<,+1]$
		are definable in $\FO_m^2[<]$.
	\end{proposition}
	\begin{proof}[Sketch of proof]
		Let $L$ be a language definable in $\FO_{m}^2[<,+1]$ and assume that it has a neutral letter $c$.
			Thanks to Theorem~\ref{theo:ef}, there exist integers $s$ and $k\leq m$  such that
			Spoiler has a winning strategy for the
		two-pebble game with $s$ rounds and $k$ alternations on $(u,v)$, with $(u,v)\in L\times L^\text{c}$.
		We construct two words $u'$ and $v'$ by inserting $2s$ letters $c$ between
		each position (including the beginning and the end of the words).
		As $c$ is a neutral letter, we have $(u',v')\in L\times L^\text{c}$ and
		therefore  Spoiler has a winning strategy for the
		two-pebble game with $s$ rounds and $k$ alternations. Remark that the successor relation
		on $(u',v')$ is useless since the non-neutral letters are not reachable from each other in less
		that $s$ rounds. Therefore one can translate the Spoiler's wining
		strategy on $(u',v')$ on a wining strategy that does not use the successor relation. This
		wining strategy can then be translated in a wining strategy on $(u,v)$.
		We then conclude thank to Theorem~\ref{theo:ef}.
	\end{proof}
\section{Main Result}\label{section:main}
	We now investigate the Crane Beach conjecture in the specific case of $\FO^2$ equipped
	with numerical predicates of finite degree.  Throughout this section, we  borrow from
	the vocabulary of   graph theory in order to express properties on the structure of numerical predicates.
	Indeed, a binary numerical predicate can be understood as a family of graphs.
	Furthermore, if the predicate is uniform, it can be viewed as a single infinite graph
	where the set of vertices is $\N$.
	Let  $P$ be a uniform numerical predicate.
	The \emph{degree} of a position
		$k$ for~$P$, denoted by $d_P(k)$, is the size of the
		 set of all integers connected to $k$ \emph{via} $P$. More formally
		$$d_P(k) = |\left\{ j \mid (k,j)\in P \text{ or }(j,k)\in P\right\}|\enspace.$$

	The notion of locality is one of the most
	effective tools for using the \EF games. One way of
	introducing locality is to restrict the degree of the signature.
	A uniform binary predicate $P$ has a \emph{finite degree} if all positions have a finite degree.
	We denote by $\cF$ the class of binary uniform finite-degree predicates.\\

	 \noindent\textbf{Examples:}
		\begin{itemize}
			\item The predicate $kx=y,$ $x^k=y, \ldots$ as well as the graph of
			any strictly growing function.
			\item The translated Bit predicate which is true in $(x,y)$ if the $(y-x)^\text{th}$ bit of $x$ is a one.

		\end{itemize}
	 \noindent\textbf{ Example of nonfinite-degree predicates:}
	\begin{itemize}
			\item The linear ordering.
			\item The Bit predicate which is true of $(x,y)$ if the $y^\text{th}$ bit of $x$ is a one.
			\item The $\bitadd_0$ predicate.
	\end{itemize}
	Predicates of finite degree do not include by definition uniform monadic predicates.
	However, all uniform monadic predicates can be encoded as  predicates of finite degree.
	If $P$ is monadic and uniform then
	$Q = \left\{ (x,x) \mid x\in P\right\}$ is a finite-degree predicate.

	The next theorem states that the Crane Beach conjecture for $\FO^2[\arb]$ reduces
	to solving the Crane Beach conjecture for the order, the $\bitadd_0$ predicate and the class
	of finite-degree predicates. The proof of this theorem is an adaptation of a circuit-version
	of a similar result~\cite{KM}. Because of the lack of space, the proof of this theorem
	is omitted.
	\begin{theorem}\label{thm:neut}
		Any language with neutral letter definable in $\FO^2[\arb]$ is definable
		in ${\FO^2[<,\cF,\bitadd_0]}$.
	\end{theorem}
		\begin{proof}
		Let $\varphi$ be a formula of $\FO^2[\arb]$ defining a language with a neutral letter.
		We are going to encode the behaviour of $\varphi$ on a word of size $n$ inside the segment
		$\{2^{i-1},\ldots,2^{i}-1\}$ where $i =\floor{\log{n}}$.  We separate the segment $\{0,\ldots,2^{i+1}-1\}$
		in four disjoint segments of size $2^{i-1}$ as follows:
		\begin{align*}
		\{0,\ldots,2^{i+1}-1\}= &\overbrace{\{0,\ldots,2^{i-1}-1\}}^{E_{-1}}\cup \overbrace{\{2^{i-1},\ldots,2^{i}-1\}}^{E_0}\\
		&\cup \underbrace{\{2^{i},\ldots,2^i+2^{i-1}-1\}}_{E_1}\cup \underbrace{\{2^i+2^{i-1},\ldots,2^{i+1}-1\}}_{E_2}\enspace.
		\end{align*}
		From now on, $i$ is fixed. 	Let $x_k$ denotes the position $x+k2^{i-1}$, where $x\in E_0$.
		Observe that:
		\begin{align*}
		\exists x\ \psi(x) &\equiv \bigvee_{k=-1}^2\exists x \ (x\in E_0\land \psi(x_k))\\
		\forall x\ \psi(x) &\equiv \bigwedge_{k=-1}^2\forall x\ (x\in E_0\to \psi(x_k))\enspace,
		\end{align*}
		where $\psi$ is any formula of $\FO^2[\arb]$.
		The variable $x$ of the formulas on the right are relativized to the segment $E_0$.
		Note that the formula $x\in E_0$ requires to use the linear order as well as the monadic predicate
		$\{2^n\mid n\in \N\}$. This formula express that there exists exactly one power of $2$ between
		$x$ and
		the end of the word. Formally:
		\begin{align*}
		x\in E_0 \equiv &\Big(\exists y\ y > x\land y\in \{2^n\mid n\in \N\}\Big) \land \\
		&\Big(
		\forall y\ \big(y > x\land y\in \{2^n\mid n\in \N\}\big) \to \big(\forall x\ x >y\to x\not\in \{2^n\mid n\in \N\}\big)\Big)
		\end{align*}
		We perform this transformation on all the quantifications of $\varphi$ and obtain a new formula
		where all variables are relativized to the segment $E_0$.

		It is quiet possible that some variable $x_k$ encodes a position that
		is larger than the size of the input word. These positions will be handled
		by considering that they are labelled by neutral letters.
		We introduce two more predicates in $\cF$:
		\begin{enumerate}
			\item The predicate $\bitadd_{10}$ defined as follows.
			The predicate $\bitadd_{10}$  is true of $x$ and $y$
			if the binary representation of $y$ is obtained by replacing the most significant bit of $x$ by $10$.
			 More formally
			$$\bitadd_{10} =\{ (x,x+2^i)\mid x\in \N, \text{ and }i=\floor{\log(x)}\text \}\enspace.$$
			\item The predicate $\bitadd_{11}$ defined as follows.
			The predicate $\bitadd_{11}$  is true of $x$ and $y$
			if the binary representation of $y$ is obtained by replacing the most significant bit of $x$ by $11$.
			 More formally
			$$\bitadd_{11} =\{ (x,x+2^{i+1})\mid x\in \N, \text{ and }i=\floor{\log(x)}\text \}\enspace.$$
		\end{enumerate}
		Now we rewrite the atoms as follows:
		\begin{itemize}
			\item The atom $P(x_k,y_\ell)$ (resp. $P(x_k)$), with $P=(P_n)$ an arbitrary numerical predicate
			is replaced by the atom
			$Q_{k,\ell}(x,y)$ (resp. $Q_{k}(x)$) where
			$Q_{k,\ell}\in\cF$ are defined as follows:
			\begin{align*}
			Q_{k,\ell} = &\left\{ (x,y) \mid 2^i\leq x,y < 2^{i+1}\text{ for some }i\in \N\text{
			 and }\big(x+k2^{i-1},y+\ell 2^{i-1}\big)\in P_{2^{i+1}-1}\right\}\\
			Q_{k} = &\left\{ x \mid 2^i\leq x < 2^{i+1}\text{ for some }i\in \N\text{
			 and }x+k2^{i-1}\in P_{2^{i+1}-1}\right\}\enspace.
 			\end{align*}
			\item The atoms $\mathbf{a}(x_0)$ is replaced by $\mathbf{a}(x)$.
			\item For $k\neq 0$ and if $a$ is a non-neutral letter, then
			the atom $\mathbf{a}(x_k)$ is replaced by the formula:
			$$\exists y\  \bitadd_{z_k}(x,y)\land \mathbf{a}(y)\enspace,$$
			where $z_{-1}=0$, $z_{1}=10$ and $z_{2}=11$.
			\item For $k\neq 0$ and if $a$ is a neutral letter, then
			the atom $\mathbf{a}(x_k)$ is replaced by the formula:
			$$\big(\exists y\  \bitadd_{z_k}(x,y)\land \mathbf{a}(y)\big)\lor
			\big(\forall y\ \lnot \bitadd_{z_k}(x,y)\big)\enspace,$$
			where $z_{-1}=0$, $z_{1}=10$ and $z_{2}=11$.
			Intuitively, if the position encoded by $x_k$ is bigger than the size of the input word,
			it is considered as a neutral letter.
		\end{itemize}
		Remark that all the extra indices have been removed during this
		process and we have now a formula of the fragment $\FO^2[<,\cF,\bitadd_0]$. Furthermore, each
		step respect the semantic of $\varphi$, so this new formula defines the same language, which concludes
		the proof.
	\end{proof}
	We believe that
	this last theorem does not hold without the neutral-letter hypothesis. For instance, the language
	$\{u\overline{u}\mid u\in A^*\}$, where $\overline{u}$ is the reversal image of $u$, is definable in
	$\FO^2[x+y=\max]$ but we conjecture that it is not definable by using only uniform predicates, and in particular, using predicates
	in the signature $[<,\cF,\bitadd_0]$.

	We now focus on the signature $[<,\cF]$.
	To solve this problem, we will use the \emph{locality} of the class $\cF$.
		Locality is an effective tool which allows us to obtain numerous results of non-definability
		with the help of the \EF games. Unfortunately, as soon as the order is present in the signature,
		it is no longer possible to use locality results and the absence of the order
		makes the fragment far less expressive.
		We are going to show that it is possible to add the order
		whilst conserving a form of locality when the other predicates are of finite degree.

	\begin{theorem}[Main Theorem]\label{thm:deg_fin}
		Let $m\geq 0$. Any language with a neutral letter definable in  $\FO_m^2[<,\cF]$
		is definable in $\FO_m^2[<]$.
	\end{theorem}
	We immediately obtain the following corollary.
	\begin{corollary}
		Any language with a neutral letter definable in  $\FO^2[<,\cF]$
		is definable in $\FO^2[<]$.
	\end{corollary}

	This theorem states the uselessness of finite-degree predicates for defining languages
	with a neutral letter in two-variable logic. More precisely, they do not even improve
	the logical complexity of the languages. 	Therefore, we
	immediately deduce the strictness of this hierarchy.
	Indeed, we mainly use the known facts that $\FO_m^2[<]$ is a strict hierarchy (see~\cite{IW09})
		and that each layer is stable by inverse  image of
		morphisms. This  latter fact is a requirement for having an equational description
		as given in the article~\cite{KS12}. Then, it is sufficient to take
		the inverse image of a language $L$ that separates $\FO_{m+1}^2[<]$ from $\FO_m^2[<]$
		by a morphism that maps a letter which is not in the alphabet of $L$ to the empty word.
	\section{Proof of the main theorem}\label{section:proof}
	 The principal ingredients are a notion of \emph{locality}, the \EF games and Ramsey's Theorem.
	For the remaining of the proof we fix $P^1,\ldots, P^t$ as predicates in $\cF$.
	Our objective is to prove that for any language $L$ with a neutral letter definable in
	${\FO_m^2[<,P^1,\ldots, P^t]}$, there exists $s$ such that for
	every words  $u\in L$ and $v\not \in L$, Spoiler has a winning strategy for the
	\EF game   with two pebbles, $s$ rounds and $m$ alternations on $(u,v)$ and over
	the signature 	$\{<,+1\}$. The proof is decomposed as follows.
	\begin{conditions}
		\item First, we introduce the notion of a \emph{position’s neighbourhood}.
		\item Then, we define an equivalence relation
		between triples of disjoint neighbourhoods, which will allow us
		to define the different roles that these triples could play throughout the course of the game.
		\item  We then extract triples of so-called \emph{well-typed}
		 positions, with the help of Ramsey's Theorem for $3$-hypergraphs.
		\item Finally, we will inductively construct  a winning strategy for Spoiler over the signature
		$\{<,+1\}$ that uses at most $s$ rounds and $m$ alternations.
		 Proposition~\ref{prop:succ} allows us to conclude.
	\end{conditions}

		Let $E\subseteq \N^2$ be defined by
			$\{x,y\}\in E$ if, and only if,  $x$ and $y$
			are two positions connected by one of the predicates.
			More precisely, $\{x,y\}\in E$ if and only if
			$$P^1(x,y)\lor P^1(y,x)\lor \cdots \lor P^t(x,y)\lor P^t(y,x)\enspace.$$

		The graph $(\N,E)$ is the graph behind our reasoning.
		As each predicate is of finite degree, the graph $(\N,E)$
		is also of finite degree. From this point on, we assume that the integer $s$
		(the number of rounds in the game) is fixed.

	 \subsection{Definition of neighbourhood}
	 For an integer~$i$, the usual notion of $r$-neighbourhood is defined as the
	set of integers at distance~$r$ from~$i$ in $(\N, E)$.  It captures the
	intuition that two integers with similar $r$-neighbourhoods cannot be
	distinguished in~$r$ applications of the predicates.  Adding linear order to
	the predicates, any element between two given integers is \emph{connected} by the order.  Our
	specialized notion of neighbourhood thus distinguishes between the linear
	order and the other predicates; to this end, let us  first introduce the
	\emph{closure} of a finite set $F \subseteq \mathbb{N}$ as:
 		$$\Cl{F} = \{\min F,\min F+1,\ldots,\max F\}\enspace.$$
	Then, intuitively combining at each step the use of the predicates and
	that of the order, we define the $0$-\emph{neighbourhood} of $i \in \mathbb{N}$ as:
	$$V(i,0)= \Cl{\{i\}\cup \bigcup_{\substack{k'\leq i \leq k\\ \{k',k\}\in E}}\{k',k\}}\enspace.$$
		and, inductively, the $(r+1)$-\emph{neighbourhood} of $i \in \mathbb{N}$ as:
		$$V(i,r+1)= \Cl{\bigcup_{j\in V(i,0)}V(j,r)}\enspace.$$
		Less formally, the $0$-neighbourhood of $i$ is the set of positions $j$ such that by moving a pebble
		inside this set
		it is possible to jump over $i$. We obtain immediatly that
		 $V(i,r)\subseteq V(i,r+1).$
		 \begin{lemma}\label{lem:tech_finite_deg1}
			For all integers $i$ and $k$,  $V(i,k)$ is finite.
		 \end{lemma}
		 \begin{proof}
		 	We prove this result by induction on $k$.
		 	In this proof, we denote by $E_j$  the set of neighbours of $j$ in the
		 	graph $(\N,E)$ for any $j\in \N$.
		 	We also remark that for $P$ a finite set, $\Cl{P}$ is also finite.
		 	\begin{itemize}
		 		\item First, the $0$-neighbourhood of $i$ is finite. For any $0\leq j\leq i$,
		 		the set $E_j$ is finite and we set
		 		$$m = \max{\bigcup_{j=0}^i E_j}\enspace.$$
		 		The $0$-neighbourhood of $i$ is then included in the segment $\{0,\ldots,m\}$
		 		which is finite.
		 		\item We assume that for any $j$, the $r$-neighbourhood of $j$ is finite and
		 		we show that the $(r+1)$-neighbourhood of  $i$
		 		finite as well.
		 		By definition
	 				$$V(i,r+1)= \Cl{\bigcup_{j\in V(i,0)}V(j,r)}\enspace.$$
				Then by induction hypothesis $V(i,0)$ and $V(j,r)$ are finite sets.
				Finally the set
				$$E=\bigcup_{j\in V(i,0)}V(j,r)$$ is a finite union of finite sets.
				This concludes the proof.
		 	\end{itemize}
		 \end{proof}
		 We now define the function
		 $g_s\colon\N\to \N$  by ${g_s(i) = \min{V(i,s)}}$.
		 \begin{lemma}
		 	 We have $\lim_i g_s(i) = +\infty$.
		 \end{lemma}
		 \begin{proof}
		 	Suppose for a contradiction that there exist
		 	$M\in \N$ and $I\subseteq \N$
		 	of infinite size such that for any integer $i\in I$, $g_s(i)\leq M$.
		 	Since $I$ is infinite, there exist $n\leq M$ and a set $I'\subseteq I$
		 	of infinite size such that
		 	for every integers $i\in I'$, we have  $g_s(i) = n$.
		 	Thanks to Lemma~\ref{lem:tech_finite_deg2},
		 	$$I'\subseteq V(n,s)\enspace.$$ A contradiction arises since the set
		 	$I'$ is infinite and the set $V(n,s)$ is finite, which concludes this proof.
		 \end{proof}

		From this we immediately deduce the following corollary,
		which establishes the possibility of obtaining an arbitrarily large number of
		neighbourhoods that do not overlap.
		\begin{corollary}\label{cor:extract}
		 	For any integer $p$, there exists $X\subseteq \N$ of size $p$
		 	such that for any $i,j\in X$, the $s$-neighbourhood of  $i$ and $j$ are disjoint and
		 	separated by at least one integer.
		 \end{corollary}

		 An $s$-\emph{extraction} is a set of integers,
		 such that  their $s$-neighbourhoods are disjoint and separated by at least one integer.
		 In short, they must be in accordance with the conditions of Corollary~\ref{cor:extract}.

		\subsection{An equivalence  relation for triples}

		We now introduce a notion of \emph{similarity} for the triples of neighbourhoods
		taken from the \EF two-pebble game. Let $(\pneg{i},i,\ppos{i})$ be a
		triple of integers which is an $s$-extraction.
 		More precisely, this triple satisfies that
 			\begin{enumerate}
 				\item $\pneg{i}<i<\ppos{i}$,
 				\item their $s$-neighbourhoods are disjoint and have at least one element between them.
 			\end{enumerate}
 			According to Corollary~\ref{cor:extract}, such a triple exists.
 			We set $J_s(i,\ppos{i})$ as the interval between the minimal position of the $s$-neighbourhood of
 			$i$ and minimal  position of the $s$-neighbourhood of $\ppos{i}$.
 			More formally,
		 	$$J_{s}(i,\ppos{i}) =\{\min{V(i,s)},\ldots,\min{V(\ppos{i},s)}-1\}\enspace.$$
			We also set $I_{(r,s)}( \pneg{i},\ppos{i})$ the interval in-between the maximal position
		 	of the $(s-r)$-neighbourhood of $\pneg{i}$ and the minimal position of the $(s-r)$-neighbourhood
		 	of $k$.
		 	More formally
		 	\begin{align*}
		 	I_{(r,s)}(\pneg{i},\ppos{i})& =  \{\max{V(\pneg{i},s-r)}+1,\ldots,\min{V(\ppos{i},s-r)}-1\}\enspace.
		 	\end{align*}
		 	These notations are illustrated in Figure~\ref{dess_preuve2}.
			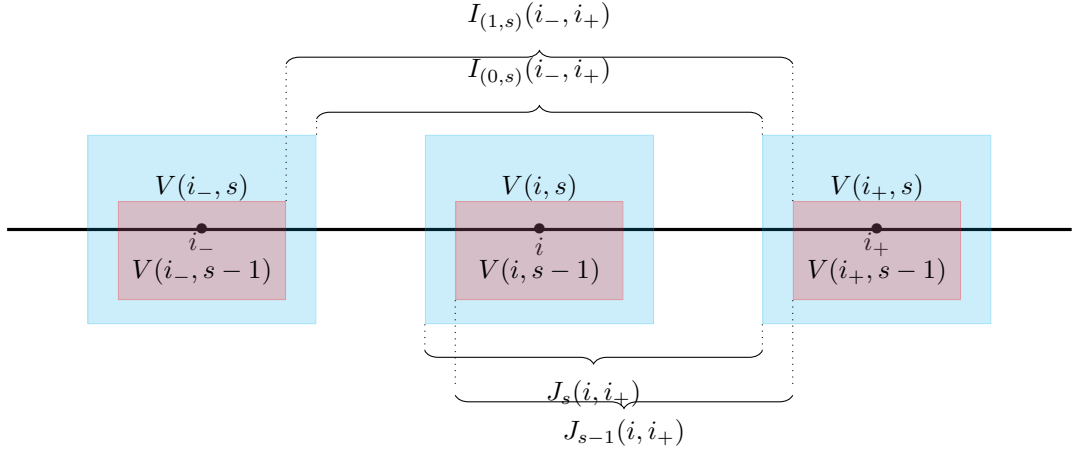
\begin{figure*}[h]
				\centering

\begin{tikzpicture}[scale=0.7]

\node (q0){};
\node (q1)[right=14of q0] {};
\path[-, draw, very thick]  (q0) -- (q1)  node (mid) [midway] {$\bullet$} ;
\node (jp)[left= 4 of mid] {$\bullet$};
\node (j)[right= 4 of mid] {$\bullet$};
\node at (mid.south) {$i$};
\node at (jp.south) {$\pneg{i}$};
\node at (j.south) {$\ppos{i}$};

\node (vj) at ($(j)+(0,0)$) [rectangle, fill, draw,opacity = 0.2, cyan,minimum  width=3cm, minimum height=2.5cm, ] {};
\node (vj2) at ($ (j) + (0,-0.4)$) [rectangle, fill, draw,opacity = 0.2, red,minimum  width=2.2cm, minimum height=1.3cm] {};

\node (vjp) at ($(jp)+(0,0.)$) [rectangle, fill, draw,opacity = 0.2, cyan,minimum  width=3cm, minimum height=2.5cm, ] {};
\node (vjp2) at ($ (jp) + (0,-0.4)$) [rectangle, fill, draw,opacity = 0.2, red,minimum  width=2.2cm, minimum height=1.3cm] {};

\node (vi) at ($(mid)+(0,0.)$) [rectangle, fill, draw,opacity = 0.2, cyan,minimum  width=3cm, minimum height=2.5cm, ] {};
\node (vi2) at ($ (mid) + (0,-0.4)$) [rectangle, fill, draw,opacity = 0.2, red,minimum  width=2.2cm, minimum height=1.3cm] {};

\node (dot0) at (vjp2.north east) {};
\node (dot1) at (vj2.north west) {};
\draw [dotted] ($(dot0)+(0,0)$) -- ($(dot0)+(0,2.6)$); 
\draw [dotted] ($(dot1)+(0,0)$) -- ($(dot1)+(0,2.6)$); 
\node (dot2) at (vi.south west) {};
\node (dot3) at (vj.south west) {};
\draw [dotted] ($(dot2)+(0,0)$) -- ($(dot2)+(0,-0.5)$); 
\draw [dotted] ($(dot3)+(0,0)$) -- ($(dot3)+(0,-0.5)$); 

\draw [dotted] (vi2.south west) -- ($(vi2.south west)+(0,-1.6)$); 
\draw [dotted] (vj2.south west) -- ($(vj2.south west)+(0,-1.6)$); 

\node (dot4) at (vjp.north east) {};
\node (dot5) at (vj.north west) {};
\draw [dotted] ($(dot4)+(0,0)$) -- ($(dot4)+(0,0.3)$); 
\draw [dotted] ($(dot5)+(0,0)$) -- ($(dot5)+(0,0.3)$); 
 \draw [decorate, decoration={brace,amplitude=5}] ($(dot0)+(0,2.6)$) -- ($(dot1)+(0,2.6)$)node [ above=2ex,   pos =0.5] {$I_{(1,s)}(\pneg{i},\ppos{i})$};

\draw [decorate, decoration={brace,mirror ,amplitude=5}] ($(vi2.south west)+(0,-1.8)$) -- ($(vj2.south west)+(0,-1.8)$)node [ below=1.5ex,  text height = 1.5ex,   pos =0.5] {$J_{s-1}(i,\ppos{i})$};
\draw [decorate, decoration={brace,mirror ,amplitude=5}] ($(dot2)+(0,-0.5)$)-- ($(dot3)+(0,-0.5)$) node [ below=2ex,  text height = 1.5ex,   pos =0.5] {$J_{s}(i,\ppos{i})$};

 \draw [decorate, decoration={brace, ,amplitude=5}]($(dot4)+(0,0.3)$) --($(dot5)+(0,0.3)$) node [ above=2ex,  text height = 1.5ex,   pos =0.5] {$I_{(0,s)}(\pneg{i},\ppos{i})$};

\node (labi) at ($ (mid) + (0,0.8)$) {$V(i,s)$};
\node (labi2) at ($ (mid) + (0,-.8)$) {$V(i,s-1)$};
\node (labjp) at ($ (jp) + (0,0.8)$) {$V(\pneg{i},s)$};
\node (labjp2) at ($ (jp) + (0,-.8)$) {$V(\pneg{i},s-1)$};

\node (labj) at ($ (j) + (0,0.8)$) {$V(\ppos{i},s)$};
\node (labj2) at ($ (j) + (0,-.8)$) {$V(\ppos{i},s-1)$};

\end{tikzpicture}
				\caption{Neighbourhoods and segments.\label{dess_preuve2}}
				\vspace{1em}
			\end{figure*}
			Let us take two triples $(\pneg{i},i,\ppos{i})$ and $(\pneg{j},j,\ppos{j})$
			 which form two $s$-extractions with $\pneg{i}<i<\ppos{i}$ and $\pneg{j}<j<\ppos{j}$.
			These two triples of integers are \emph{equivalent}
			 if two two-pebble \emph{constrained  games} are similar.
			 We  define two different notions of \emph{constrained games} that differ only in
			 their starting sets.
			These games only use two pebbles which are confined, at the $r^\text{th}$ round, to the intervals
			$$I_{(r,s)}(\pneg{i},\ppos{i}) \text{ and }
			I_{(r,s)}(\pneg{j},\ppos{j})\enspace.$$

			For the first game, the first pebble must be placed for both Spoiler and Duplicator in the sets
			 $J_s(i,\ppos{i}) \text{ and }J_s(j,\ppos{j}).$
			For the second game the first pebble is placed by Spoiler and Duplicator in the sets
			$V(i,s) \text{ and }V(j,s).$
			If Duplicator wins these two games we can state that these two triples are {equivalent},
			which we denote as
			$(\pneg{i},i,\ppos{i})\sim_s (\pneg{j},j,\ppos{j}).$

			We now introduce formally this definition. We say that
			$(\pneg{i},i,\ppos{i})\sim_s (\pneg{j},j,\ppos{j})$
			if for all $s'\leq s$
			Duplicator wins the two following games.
		 	They are two-pebble games with  $s'$ rounds and $s'$ alternation
		 	(we consider that Spoiler
		 	may alternate as much as he wishes between the two words)
		 	on the  signature $\{<,P^1,\ldots,P^t\}$, and all positions
		 	are labelled by the same letter $a$ with the exception of positions $i$ and $j$
		 	which are labelled by the same letter $b$ distinct from $a$. Here is the formal description of the two games:
		 	\begin{conditions}
		 		\item For the first game, the first pebble of  Spoiler and the first pebble of Duplicator
		 		are constrained to the set
		 		$J_{s'}(i,\ppos{i})$ and $J_{s'}(j,\ppos{j}).$
		 		At the $r^\text{th}$ round, the players are constrained to choose positions in the
		 		sets
			 	$I_{(r,s')}(\pneg{i},\ppos{i})$ and $I_{(r,s')}(\pneg{j},\ppos{j}).$
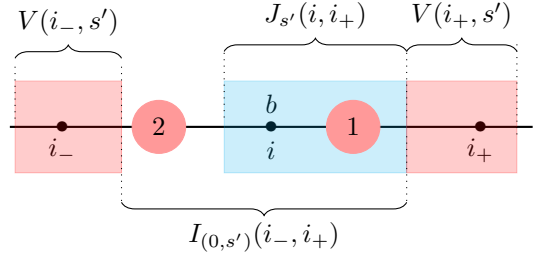
				\begin{center}
			 	\usetikzlibrary{calc,automata,arrows,chains,matrix,positioning,scopes,decorations.pathreplacing}

\begin{tikzpicture}[scale = 1,every node/.style={transform shape}]
\def\wlength{7}
\node (u){};
\node (v) [below= 2of u] {};

\path[-, draw, thick]  (u) -- ($(u)+(\wlength,0)$) node (jp) [pos = 0.1] {$\bullet$} 
node (i) [pos = 0.5] {$\bullet$}  node (j) [pos = 0.9] {$\bullet$};

\node (jpl) at ($(jp.south)+(0,-.1)$) {$\pneg{i}$};

\node (il) at ($(i.south)+(0,-.1)$) {$i$};
\node (il2) at ($(i.north)+(0,.1)$) {$b$};
\node (jl) at ($(j.south)+(0,-.1)$) {$\ppos{i}$};

\node (dot0) at ($(i.south west)+(2,-0.4)$) {};

\draw[ cyan, opacity = 0.2,fill] ($(i.north west)+(-0.4,0.4)$) rectangle  (dot0);
\draw[ red, opacity = 0.2,fill] ($(jp.north west)+(-0.4,0.4)$) rectangle  ($(jp.south west)+(1,-0.4)$);
\draw[ red, opacity = 0.2,fill] (dot0) rectangle  ($(jp.north west)+(6.2,0.4)$);

\node (spoiler1) [right = 0.5of i] [draw,circle , fill, red!40!white,opacity=1, label, minimum size = 20 ]{};
\node at (spoiler1) {$1$};
\node (spoiler2) [left =0.9 of i] [draw,circle , fill, red!40!white,opacity=1, label, minimum size = 20 ]{};
\node at (spoiler2) {$2$};
\node (ndot) at  ($(i.west)+(-0.4,0)$) {};
\node (ndot1) at  ($(i.west)+(2,0)$) {};
\node (ndot3) at  ($(jp.west)+(1,0)$) {};
\draw [decorate, decoration={brace,,amplitude=5}] ($(ndot)+(0,1)$) -- ($(ndot1) +(0,1)$) node [ above=0.5em,  text height = 1.5ex,   pos =0.5] {$J_{s'}(i,\ppos{i})$};
\draw [decorate, decoration={brace,,amplitude=5}] ($(jp.north west)+(-0.4,0.7)$)-- ($(jp.north west)+(1,.7)$) node [ above=0.5em,  text height = 1.5ex,   pos =0.5] {$V(\pneg{i},s')$};
\draw[dotted] ($(jp.north west)+(-0.4,-0.1)$) -- +(0,0.7);
\draw[dotted] ($(jp.north west)+(1,-0.1)$) -- +(0,0.7);
\draw [decorate, decoration={brace,,amplitude=5}] ($(ndot1) +(0,1)$) -- +(1.45,0) node [ above=0.5em,  text height = 1.5ex,   pos =0.5] {$V(\ppos{i},s')$};
\draw[dotted] ($(jp.west)+(6.2,-0.1)$) -- +(0,1.2);

\draw[dotted] (ndot) -- +(0,1);
\draw[dotted] (ndot1) -- +(0,1);

\draw [decorate, decoration={brace,mirror,amplitude=5}] ($(ndot3)+(0,-1)$) -- ($(ndot1) +(0,-1)$) node [ below=0.5em,  text height = 1.5ex,   pos =0.5] {$I_{(0,s')}(\pneg{i},\ppos{i})$};
\draw[dotted] (ndot3) -- +(0,-1);
\draw[dotted] (ndot1) -- +(0,-1);

%
%

\end{tikzpicture}
			 	\end{center}
			 	\item For the second game, the Spoiler's first pebble and the first pebble of Duplicator
		 		are constrained to $V(i,s')$ and $V(j,s')$.
		 		At the $r^\text{th}$ round, the players are constrained to play in the sets
			 	$I_{(r,s')}(\pneg{i},\ppos{i})$ and $I_{(r,s')}(\pneg{j},\ppos{j}).$
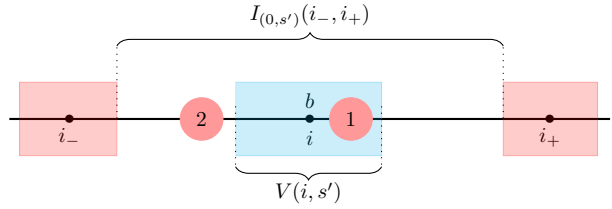
			 	\begin{center}
			 	\usetikzlibrary{calc,automata,arrows,chains,matrix,positioning,scopes,decorations.pathreplacing}

\begin{tikzpicture}[scale = 0.8,every node/.style={transform shape}]
\def\wlength{10}
\node (u){};
\node (v) [below= 2of u] {};

\path[-, draw, thick]  (u) -- ($(u)+(\wlength,0)$) node (jp) [pos = 0.1] {$\bullet$} 
node (i) [pos = 0.5] {$\bullet$}  node (j) [pos = 0.9] {$\bullet$};

\node (jpl) at ($(jp.south)+(0,-.1)$) {$\pneg{i}$};
\node (il) at ($(i.south)+(0,-.1)$) {$i$};
\node (il2) at ($(i.north)+(0,.1)$) {$b$};
\node (jl) at ($(j.south)+(0,-.1)$) {$\ppos{i}$};

\node (dot0) at ($(i.south west)+(3,-0.4)$) {};

\draw[ cyan, opacity = 0.2,fill] ($(i.north west)+(-1,0.4)$) rectangle ($(i.south west)+(1.4,-0.4)$);
\draw[ red, opacity = 0.2,fill] ($(jp.north west)+(-0.6,0.4)$) rectangle  ($(jp.south west)+(,-0.4)$);
\draw[ red, opacity = 0.2,fill] ($(j.north west)+(-0.55,0.4)$) rectangle  ($(j.south west)+(,-0.4)$);

\node (spoiler1) [right = 0.1 of i] [draw,circle , fill, red!40!white,opacity=1, label, minimum size = 20 ]{};
\node at (spoiler1) {$1$};
\node (spoiler2) [left =1.2 of i] [draw,circle , fill, red!40!white,opacity=1, label, minimum size = 20 ]{};
\node at (spoiler2) {$2$};

\draw [decorate, decoration={brace,,amplitude=5}] ($(jp.north west)+(,1.)$) --  ($(j.north west)+(-0.55,1)$) node [ above=0.5em,  text height = 1.5ex,   pos =0.5] {$I_{(0,s')}(\pneg{i},\ppos{i})$};
\draw[dotted] ($(jp.north west)+(,-0.2)$) -- +(0,1.2); 
\draw[dotted] ($(j.north west)+(-0.55,-0.2)$) -- +(0,1.2); 

\draw [decorate, decoration={brace,mirror,amplitude=5}] ($(i.north west)+(-1,-1)$) --  ($(i.north west)+(1.4,-1)$) node [ below=0.5em,  text height = 1.5ex,   pos =0.5] {$V(i,s')$};
\draw[dotted] ($(i.north west)+(-1,0)$) -- +(0,-1); 
\draw[dotted] ($(i.north west)+(1.4,0)$) -- +(0,-1); 


\end{tikzpicture}
				\end{center}
		 	\end{conditions}
		 	We say  that positions
		 	$x\in I_{(r,s')}(\pneg{i},\ppos{i})$ and $y\in I_{(r,s')}(\pneg{j},\ppos{j})$ are \emph{locally
		 	equivalent}  if Duplicator can win the two restricted games when the pebbles are at these positions.
		 	{The property presented in the following lemma can be deduced from the definitions and
		 	will be useful later.}

			\begin{lemma}\label{lem:tech_proof}
				Let $(\pneg{i},i,\ppos{i})$ an $s$-extraction.
					For every $0\leq r\leq s$, we have the following
					$$J_{s-r}(\pneg{i},i)\cup J_{s-r}(i,\ppos{i})= V(\pneg{i},s-r)\cup I_{(r,s)}(\pneg{i},\ppos{i})\enspace.$$
			\end{lemma}

			We now prove that $\sim_s$  is a {finite-index equivalence} relation.
			This is a rather classical result for this type of object in finite model theory.
			We remark  that the equivalent classes can be seen as the sets of true formulae for each triple in a logic
			adapted to the two restricted games. Thus,  two triples would be equivalent if they satisfy the
			same formulae of quantifier depth  {less than} $s$.
			As the number of formulae is finite, we can easily deduce that $\sim_s$
			equivalence relation.

		 	\begin{lemma}
		 		The relation $\sim_s$ is an equivalence relation of finite index.
		 	\end{lemma}
 				We decompose this lemma in two intermediary result.
		 	\begin{lemma*}
		 		The relation $\sim_s$ is an equivalence relation.
		 	\end{lemma*}
		 	\begin{proof}
		 		The $\sim_s$ relation is clearly symmetrical and reflexive. Let $x,y$ and $z$ be triples
		 		forming an $s$-extraction such that
				$x\sim_s y$ and  $y\sim_s z$.
		 		We now show that $x\sim_s z$.
		 	    First we denote by  $S_r(x)$, $S_r(y)$ and
		 		$S_r(z)$ the authorized positions for these triples at the $r^\text{th}$ round.
		 		We are going to play the following three games simultaneously with $s'\leq s$.
		 		\begin{conditions}
		 		\item The first  on $S_{s'}(x)$ and $S_{s'}(y)$.
		 		\item The second on $S_{s'}(y)$ and $S_{s'}(z)$.
		 		\item The third on $S_{s'}(x)$ and $S_{s'}(z)$.
				\end{conditions}
				For the first two, Duplicator has a winning strategy.
				We  use it to construct a winning strategy for the third game.
				Let $r\leq s'$ and assume that Spoiler plays $s_1$ in $S_r(x)$ on the third game at round $r$.
				We simulate Spoiler’s choice by playing a position $s_1$ in $S_r(x)$ in the first game.
				Duplicator  then responds by following his winning strategy and by choosing
				a position $s_2$ in $S_2(y)$ for the first game.
				We then simulate Spoiler’s choice in position $s_2$ of $S_r(y)$ for the second game.
				Once again, Duplicator  answers with his winning strategy and chooses a  position $s_3$
				in $S_r(z)$. Finally, we  choose this  position $s_3$
				to respond to Spoiler’s choice in the third game.

		 		By following Duplicator’s strategies, we immediately deduce
		 		that Duplicator also has a strategy for the third game.

		 	\end{proof}
			We now prove that this relation has  finite index.
		 	\begin{lemma*}
		 		The equivalence relation $\sim_s$ has  finite index.
		 	\end{lemma*}
		 	\begin{proof}
		 		Let $S=(\pneg{i},i,\ppos{i})$ be an $s$-extraction with
		 		$$\pneg{i}<i<\ppos{i}\enspace.$$
			Rather than introducing an artificial notion of logic adapted to the restricted games,
			we build inductively a notion of \emph{type}
			 to prove that this relation of equivalence has a finite index.
		 		We set $r\ttau_S(x)$ to be the $r$-type
		 		of a position $x$ in $I_{(r,s)}(\pneg{i},\ppos{i})$, defined as follows.
		 		\begin{itemize}
		 			\item For all $x$ in $I_{(0,s)}(\pneg{i},\ppos{i})$
		 			we set $0\ttau_S(x)$ the $(t+2)$ tuple of the binary values of predicates in the signature.
		 			More formally we have
		 			  \begin{align*}
		 			 	0\ttau_S(x) =& \big(x<i,x>i,P^1(x,x),\ldots\\
		 			 	&\ldots,P^t(x,x)\big)\in \{0,1\}^{t+2}\enspace.
		 			  \end{align*}
		 			\item For $0\leq r<s$ and all $x$ in $I_{(r+1,s)}(\pneg{i},\ppos{i})$ we set
		 			$$(r+1)\ttau_S(x) = \Big\{\Big(C(x,y),r\ttau(y)\Big) \mid y\in I_{(r,s)}(\pneg{i},\ppos{i})\Big\}$$
		 			with $C(x,y)$ the binary value of predicates $P^i$ between $x$ and $y$:
		 			\begin{align*}
		 			C(x,y) &= \big(x<y,x>y,P^1(x,y),P^1(y,x),\ldots\\
		 			&\ldots,P^t(x,y),P^t(y,x)\big)\in\{0,1\}^{2t+2}\enspace.
		 			\end{align*}

		 		\end{itemize}
		 		The $s'$-type of a triple $S$ is the couple
		 		$$\Big(\{s'\ttau_S(x)\mid x\in J_{s'}(i,\ppos{i})\},\{s'\ttau_S(x)\mid x\in V(i,s)\}\Big) \enspace.$$
		 		By definition, there exists a finite number of $s'$-types of positions and so
		 		a finite number of $s'$-types of triples. By definition of the notion of type
		 		and by immediate induction we obtain that if
		 		$(\pneg{i},i,\ppos{i})$ and $(\pneg{j},j,\ppos{j})$ have the same $s'$-type then
		 		$(\pneg{i},i,\ppos{i})\sim_s (\pneg{j},j,\ppos{j})$. Therefore $\sim_s$ has finite index.
		 	\end{proof}

				Ramsey’s Theorem is a combinatorial result of graph theory often used in finite model theory.
				Here we use a version adapted to \emph{hypergraphs}. We introduce it in the context of triples,
				which is a direct reformulation of the 3-\emph{hypergraphs} variant.
				This theorem establishes that for every large {hypergraph} with coloured edges,
				it is possible to extract a sufficiently large monochrome sub-hypergraph.
				This theorem  allows us to find an arbitrarily large set of
				triples which are all pairwise equivalent for the $\sim_s$ relation.
				For a set $E$, we denote by $\St(E)$ the set of pairwise disjoint triples of $E$.

		 	\begin{theorem}[Ramsey's Theorem for $3$-hypergraphs~\cite{Ramsey30}]
		 	 Let $c$ be an integer. For any integer $p$ there exists an integer
		 	 $n$ such that for any set $S$ of size $n$ and any function
		 	 $h\colon\St(S)\to \{1,\ldots,c\}$ there exists a set $F\subseteq S$ of size
		 	 $p$ such that $h$ is constant on  $\St(F)$.
		 	\end{theorem}

		 	A {well-typed} $s$-extraction is a set $X$ that is an $s$-extraction and
		 	such that all the triples of $X$ are equivalent for $\sim_s$.
		 	The following corollary is an immediate from Ramsey’s Theorem,
		 	in which $c$ is the number of $s$-types of triples and $h$ is the function
		 	that associates triple with their $s$-type.
		 	\begin{corollary}\label{cor:welltype}
		 		For all integers $p$ there exists a well-typed  $s$-extraction of size $p$.
		 	\end{corollary}
		 	We have now presented all of the tools necessary to present a proof of Theorem~\ref{thm:deg_fin}.

			\subsection{Core of the proof}
				Let $L$  be a language with $c$ {as a neutral letter}
				 and definable in  $\FO_m^2[<,P^1,\ldots,P^t]$.
				  According to Theorem~\ref{theo:ef}, there exists an integer $s$,
				  such that for any  words $(u,v)\in L\times L^\text{c}$,
				  Spoiler has a winning strategy for the two-pebble game with  $s$ rounds and $m$ alternations
				  for the signature ${\{<, P^1,\ldots,P^t\}}$.
				  Let $(u,v)$ in $L\times L^\text{c}$ be such a pair.
				  We now construct a strategy for Spoiler using only the order and the successor.
				  Let $p=\max(|u|,|v|)+1$. According to Corollary~\ref{cor:welltype}, there exists
				 $X=\{i_0<i_1 \cdots < i_{p}\},$
				   which is a  well-typed $s$-extraction. Let $n=\max V(i_p,s)$, and
				  let $u'$ and $v'$ be two words of length $n$ and $(f_i)_{0\leq i < |u|}, (g_i)_{0\leq i < |v|}$
				  such that:
			  \begin{itemize}
			  	\item  $i_0<f_0<f_1<\cdots<f_{|u|-1}<f_{|u|} = i_p,$ and
				$i_0<g_0<g_1<\cdots<g_{|v|-1}<g_{|v|} = i_p,$
			  	\item for all integers $i$, the positions $f_i$ and $g_i$ belong to  $X$,
			 	\item  $u'_{f_i}= u_i$, $v'_{g_i} = v_i$,
			  	$f_0=g_0$ and $f_{|u|-1} = g_{|v|-1}$,
			  	\item all unassigned positions of $u'$ and $v'$ are labelled by the letter $c$.
			  \end{itemize}
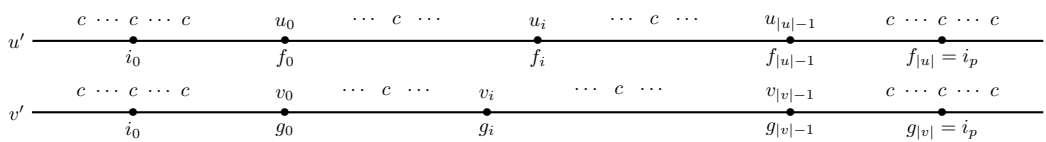
			 \begin{center}

\begin{tikzpicture}[scale = 0.75,every node/.style={transform shape}]
\def\wlength{18}
\node (up){$u'$};
\path[-, draw, thick]  (up) -- ($(up)+(\wlength,0)$) node (iu1) [pos = 0.1] {$\bullet$}  node (iu2) [pos = 0.25] {$\bullet$} node (uinv1)[pos = 0.36] {}  node (iu3) [pos = 0.5] {$\bullet$} node (uinv2)[pos = 0.615] {}  node (iu4) [pos = 0.75] {$\bullet$} node (iu5) [pos = 0.9] {$\bullet$};
\node (u1l) at ($(iu1.south)+(0,-.1)$) {$i_0$};
\node (u1l) at ($(iu1.north)+(0,.15)$) {$c\ \cdots\ c \ \cdots  \ c$};
\node (u2l) at ($(iu2.south)+(0,-.1)$) {$f_0$};
\node (u2l) at ($(iu2.north)+(0,.1)$) {$u_0$};
\node (u3l) at ($(iu3.south)+(0,-.1)$) {$f_i$};
\node (u3l) at ($(iu3.north)+(0,.1)$) {$u_i$};
\node (cu) at ($(uinv1.north)+(0,.25)$) {$\ \  \cdots \ \ c\ \  \cdots \ \   $};
\node (cu2) at ($(uinv2.north)+(0,.25)$) {$\ \  \cdots \ \ c\ \  \cdots \ \    $};
\node (u4l) at ($(iu4.south)+(0,-.1)$) {$f_{|u|-1}$};
\node (u4l) at ($(iu4.north)+(0,.1)$) {$u_{|u|-1}$};
\node (u5l) at ($(iu5.south)+(0,-.1)$) {$f_{|u|}=i_p$};
\node (u5l) at ($(iu5.north)+(0,.15)$) {$c\ \cdots\ c \ \cdots  \ c$};

\node (vp)[below = 0.75 of up]{$v'$};
\path[-, draw, thick]  (vp) -- ($(vp)+(\wlength,0)$) node (iv1) [pos = 0.1] {$\bullet$}  node (iv2) [pos = 0.25] {$\bullet$} node (vinv1)[pos = 0.35] {}  node (iv3) [pos = 0.45] {$\bullet$} node (vinv2)[pos = 0.58] {}  node (iv4) [pos = 0.75] {$\bullet$} node (iv5) [pos = 0.9] {$\bullet$};
\node (v1l) at ($(iv1.south)+(0,-.1)$) {$i_0$};
\node (v1l) at ($(iv1.north)+(0,.15)$) {$c\ \cdots\ c \ \cdots  \ c$};
\node (v2l) at ($(iv2.south)+(0,-.1)$) {$g_0$};
\node (v2l) at ($(iv2.north)+(0,.1)$) {$v_0$};
\node (v3l) at ($(iv3.south)+(0,-.1)$) {$g_i$};
\node (v3l) at ($(iv3.north)+(0,.1)$) {$v_i$};
\node (cv) at ($(vinv1.north)+(0,.25)$) {$\ \  \cdots \ \ c\ \  \cdots \ \   $};
\node (cv2) at ($(vinv2.north)+(0,.27)$) {$\ \  \cdots \ \ c\ \  \cdots \ \   $};
\node (v4l) at ($(iv4.south)+(0,-.1)$) {$g_{|v|-1}$};
\node (v4l) at ($(iv4.north)+(0,.1)$) {$v_{|v|-1}$};
\node (v5l) at ($(iv5.south)+(0,-.1)$) {$g_{|v|}=i_p$};
\node (v5l) at ($(iv5.north)+(0,.15)$) {$c\ \cdots\ c \ \cdots  \ c$};

\end{tikzpicture}
			\end{center}
			 If the words $u$ and $v$ are not of the same size,
			 then that could give us $f_i\neq g_i$.
			 The words $u'$ and $v'$
			 are nothing other than the words $u$ and $v$ after
			inserting neutral letters such that the non-neutral letters are on $X$.
			 We also require the first and last non-neutral letters to be in the exact same positions.

				As $c$ is a neutral letter, $(u',v')$ is in $L\times L^\text{c}$.
			 Therefore, Spoiler has a winning strategy for the two-pebble game over $s$-round and $m$-alternation
			   and the signature  $\{<,P^1,\ldots,P^t\}$.
			 We now have to construct Spoiler’s new strategy on $(u,v)$. In order to do so,
			 we  simulate the game on $(u',v')$  and construct \emph{via} induction
			 a winning strategy for Spoiler on $(u,v)$.
			 To achieve this step, we  exploit a back-and-forth mechanism between the game on $(u',v')$
			  and the game on $(u,v)$.
			  By following his winning strategy, Spoiler  chooses a position on $(u',v')$
			  which we  translate into a position in $(u,v)$.
			  Duplicator  then chooses a position in $(u,v)$
			  which we  translate on a position in $(u',v')$.
			  We repeat this process until Duplicator can no longer respond in $(u',v')$.
			  We must force Spoiler to  play moves that are distant from one another so that
			  his choices in $(u',v')$ lead to a winning strategy on
			 $(u,v)$.
			  If Spoiler’s new pebble is in a neighbourhood different to that of the previous pebble,
			  then by construction of the neighbourhoods, the numerical predicates, with the exception of
			  the order predicate, do not allow for a connection between the two positions; they do not
			  transmit \emph{information}.
			 In the following section we  always denote by $i_r$ (resp. $j_r$)
			 the position of the pebble played at the round  $r$ on $u$ (resp. $v$).
			  Likewise, we use $i'_r$ (resp. $j'_r$)
			  for the position of the pebble at the round $r$ on $u'$ (resp. $v'$).

			  For this construction to work, Spoiler should not win the game on $(u',v')$
			  before he wins it on $(u,v)$. This could however happen if Duplicator’s choices on $(u',v')$
			  are not pertinent.
			  We  avoid this situation by selecting locally equivalent positions, that is, positions where
			  Duplicator wins the restricted games introduced in the preceding section. Thus, Spoiler cannot win by
			  choosing moves that are close to the old pebbles. He is therefore forced to play some distant moves.

				When Spoiler plays on an \emph{extremal position} of the game on $(u',v')$,
				 Duplicator can always respond at the same position on the other word.
				 These moves therefore are of no interest in Spoiler’s strategy.
				 They are not used in the construction of the strategy of the game on $(u,v)$.
				 Each time  Spoiler makes such a move, the game on $(u,v)$ does not progress.
				 More specifically, if the game has not started, the pebbles
				 are not even placed and if the pebbles are already placed, they are not moved.

		 		We begin by describing the game’s first round, then we  inductively build  a strategy
		 		for the following rounds. For the first move, Spoiler’s winning strategy
		 		designates a position for the game on $(u',v')$.
		 		Through symmetry, we  assume that this is a position on $u'$.
		 		 We therefore distinguish two cases:
		 	 \begin{enumerate}
		 	 	\item This first move occurs within a segment of the form $J_s(f_i,f_{i+1})$
		 	 	for an integer  $0\leq i<|u|$.
		 	 	In this case, we  choose to play on the position $i$ on the game on $(u,v)$.
		 	 	Duplicator then responds in the game on $(u,v)$ by playing on $v$ at a position $j$.
		 	 	If the letter that marks $j$ is different from the one that marks $i$,
		 	 	Duplicator loses the game immediately. Otherwise, we  have to simulate
		 	 	Duplicator’s response in the game on $(u',v')$ by choosing a position in $J_s(g_j,g_{j+1})$
	 	 	 that is locally equivalent to Spoiler’s first pebble.
	 	 	 This is possible as the letters that mark $f_i$ on $u'$ and $g_j$ on $v'$ are equal, and
		 	 ${(f_{i-1},f_i,f_{i+1})\sim_s (g_{j-1},g_j,g_{j+1}).}$
		 	 	
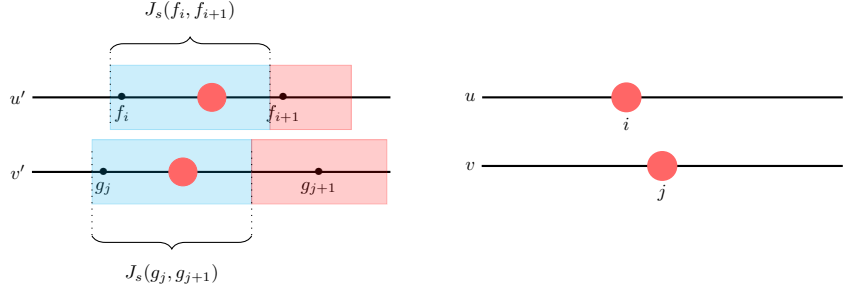
\begin{center}
		 	 	\begin{tikzpicture}[scale = 0.7,every node/.style={transform shape}]
\def\wlength{7}
\node (up){$u'$};
\node (vp) [below=0.9 of up] {$v'$};
\node (u)[right= 8 of up]{$u$};
\node (v) [below=0.9 of u] {$v$};

\path[-, draw, thick]  (up) -- ($(up)+(\wlength,0)$) node (iu) [pos = 0.25] {$\bullet$} node (iu2) [pos = 0.7] {$\bullet$};

\node (iul) at ($(iu.south)+(0,-.1)$) {$f_i$};
\node (iu2l) at ($(iu2.south)+(0,-.1)$) {$f_{i+1}$};
\node (dot0) at ($(iu.south west)+(3,-0.4)$) {};
\node (dot1) at ($(iu.south west)+(0,-0.4)$) {};
\node (dot2) at ($(iu.south west)+(5.8,-0.4)$) {};

\draw[ cyan, opacity = 0.2,fill] ($(iu.north west)+(0,0.4)$) rectangle  (dot0);
\draw[ red, opacity = 0.2,fill] ($(iu.north west)+(3,0.4)$) rectangle  ($(iu2.south west)+(1.5,-0.4)$);
\draw[dotted] (dot0) -- ($(dot0) +(0,1.5)$);
\draw[dotted] (dot1) -- ($(dot1) +(0,1.5)$);
\node (spoiler1) [right =1.2 of iu] [draw,circle , fill, red!60!white,opacity=1, label, minimum size = 15 ]{};
\draw [decorate, decoration={brace,,amplitude=5}] ($(dot1) +(0,1.5)$) -- ($(dot0) +(0,1.5)$) node [ above=3ex,  text height = 1.5ex,   pos =0.5] {$J_s(f_i,f_{i+1})$};
\path[-, draw, thick]  (u) -- ($(u)+(\wlength,0)$)  node (spoiler1p) [pos= 0.4, draw,circle , fill, red!60!white,opacity=1, label, minimum size = 15 ]{};
\node at ($(spoiler1p)+(0,-0.5)$) {$i$};
\path[-, draw, thick]  (v) -- ($(v)+(\wlength,0)$) node (Dupl1) [pos= 0.5, draw,circle , fill, red!60!white,opacity=1, label, minimum size = 15 ]{};
\node at ($(Dupl1)+(0,-0.5)$) {$j$};

\path[-, draw, thick]  (vp) -- ($(vp)+(\wlength,0)$) node (iv) [pos = 0.2] {$\bullet$} node (iv2) [pos = 0.8] {$\bullet$};

\node (ivl) at ($(iv.south)+(0,-.1)$) {$g_j$};
\node (iv2l) at ($(iv2.south)+(0,-.1)$) {$g_{j+1}$};
\node (dotv0) at ($(iv.south west)+(3,-0.4)$) {};
\node (dotv1) at ($(iv.south west)+(0,-0.4)$) {};
\node (dotv2) at ($(iv.south west)+(5.8,-0.4)$) {};

\draw[ cyan, opacity = 0.2,fill] ($(iv.north west)+(0,0.4)$) rectangle  (dotv0);
\draw[ red, opacity = 0.2,fill] ($(iv.north west)+(3,0.4)$) rectangle  ($(iv2.south west)+(1.5,-0.4)$);
\draw[dotted] ($(dotv0) +(0,1)$) -- ($(dotv0) +(0,-0.6)$);
\draw[dotted] ($(dotv1) +(0,1)$) -- ($(dotv1) +(0,-0.6)$);
\node (spoiler1) [right = of iv] [draw,circle , fill, red!60!white,opacity=1, label, minimum size = 15 ]{};
\draw [decorate, decoration={brace,,amplitude=5}] ($(dotv0) +(0,-0.6)$) -- ($(dotv1) +(0,-0.6)$)node [ below=3ex,  text height = 1.5ex,   pos =0.5] {$J_s(g_j,g_{j+1})$};

\end{tikzpicture}
		 	 	\end{center}
		 	 	\item This first move is on an extremal position, that is
				smaller than
		 	 	$\min J_s(f_0,f_1)=\min J_s(g_0,g_1)$
		 	 	or bigger than
		 	 	$\max J_s(f_{|u|-1},f_{|u|})=\max J_s(g_{|v|-1},g_{|v|}).$
		 	 	In this case, the back-and-forth process is degenerate since the game on $(u,v)$ has
		 	 	not started yet. It starts when Spoiler  plays on a non-extremal position.

				This kind of moves is not useful for Spoiler since Duplicator can only answer
				on the game on $(u',v')$ by choosing the exact same position on the other word.
				As long as Spoiler plays on these extremal positions, it is
				sufficient for Duplicator to choose the exact same position. As Spoiler
				follows a winning strategy, he eventually plays inside a segment
		 	 	$J_s(f_i,f_{i+1})$ for some
		 	 	integers $0\leq i<|u|$.
		 	 	Indeed, the extremal positions together with segments
		 	 	$J_s(f_i,f_{i+1})$ split into a partition of all positions of the word (see
		 	 	 Figure~\ref{dess_preuve2}).
		 	 	 Therefore, we can assume to be in the preceding case.
		 	 \end{enumerate}
		 	 We now explain how to construct a winning strategy for Spoiler on
		 	$(u,v)$ for the next rounds. We construct it inductively.
		 	We now assume to have played  $1\leq r<s$ rounds
		 	and that the pebbles of the preceding round are on positions
		 	 $i_r$ on $u$ (resp. $j_r$ on $v$) as well as
		 	 $i'_r$ on $u'$ (resp. $j'_r$ on $v'$).
			It is Spoiler's turn to play.
			By induction, we assume the following properties
			to be satisfied:
		 	\begin{itemize}
		 		\item If positions $i_r'$ and $j_r'$
		 		belong to
		 		$I_{(r,s)}(f_{i_r-1},f_{i_r+1})$ and to
		 		$I_{(r,s)}(g_{j_r-1},g_{j_r+1})$
		 		then they
		 		are locally equivalent for at least one of the two constrained games
		 		at $(s-r)$-rounds (see Figure~\ref{fig:dessin_preuve7}).
		 		The first constrained game corresponds to the second case, and the second
		 		constrained game corresponds to the third case.

		 		\item
		 		If this latter condition is not satisfied, then
		 		both pebbles have the exact same value, which is an extremal position on $(u',v')$.
		 		 More precisely, $i_r'=j_r'$ and either
		 		$$i_r'<\min J_{s-r}(f_0,f_1) = \min J_{s-r}(g_0,g_1)\text{ or }i_r'>\max J_{s-r}(f_{|u|-1},f_{|u|})=\max J_{s-r}(g_{|v|-1},g_{|v|})\enspace.$$
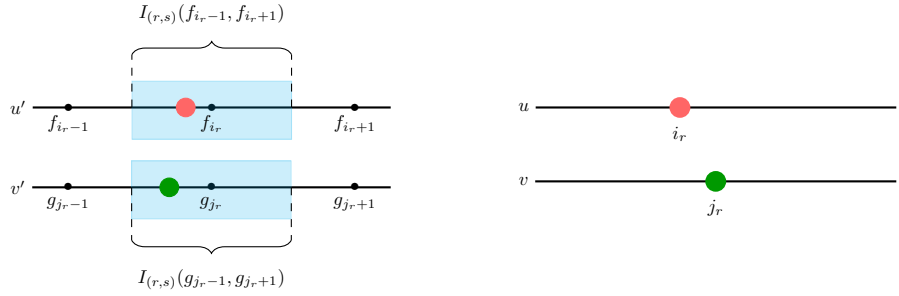
		 		\begin{center}
		 	 	\usetikzlibrary{calc,automata,arrows,chains,matrix,positioning,scopes,decorations.pathreplacing}

\begin{tikzpicture}[scale=0.7,every node/.style={transform shape}]
\def\wlength{7}5
\node (up){$u'$};
\node (vp) [below=1 of up] {$v'$};
\node (u)[right= 9 of up]{$u$};
\node (v) [below=1 of u] {$v$};

\path[-, draw, thick]  (up) -- ($(up)+(\wlength,0)$) node (iu) [pos = 0.1] {$\bullet$}  node (iu3) [pos = 0.5] {$\bullet$} node (iu2) [pos = 0.9] {$\bullet$};
\path[-, draw, thick]  (vp) -- ($(vp)+(\wlength,0)$) node (iv) [pos = 0.1] {$\bullet$}  node (iv3) [pos = 0.5] {$\bullet$} node (iv2) [pos = 0.9] {$\bullet$};

\node (dot0) at ($(iu3)+(-1.5,0.5)$) {};
\node (dot1) at ($(iv3)+(-1.5,0.5)$) {};

\draw[ cyan, opacity = 0.2,fill] ($(iu3)+(-1.5,0.5)$) rectangle  ($(iu3)+(1.5,-0.6)$);
\draw[ cyan, opacity = 0.2,fill] ($(iv3)+(-1.5,0.5)$) rectangle  ($(iv3)+(1.5,-0.6)$);
\draw[dashed]  ($(dot0)+(0,-0.5)$) --  ($(dot0)+(0,0.5)$);
\draw[dashed]  ($(dot1)+(0,-0.5)$) --  ($(dot1)+(0,-1.5)$);
\draw[dashed]  ($(dot0)+(3,-0.5)$) --  ($(dot0)+(3,0.5)$);
\draw[dashed]  ($(dot1)+(3,-0.5)$) --  ($(dot1)+(3,-1.5)$);
\draw [decorate, decoration={brace,mirror,amplitude=5}] ($(dot0)+(3,0.5)$) -- ($(dot0)+(0,0.5)$) node [ above=3ex,  text height = 1.5ex,   pos =0.5] {$I_{(r,s)}(f_{i_r-1},f_{i_r+1})$};
\draw [decorate, decoration={brace,,amplitude=5}] ($(dot1)+(3,-1.5)$) -- ($(dot1)+(0,-1.5)$) node [ below=3ex,  text height = 1.5ex,   pos =0.5] {$I_{(r,s)}(g_{j_r-1},g_{j_r+1})$};

\node (iul) at ($(iu.south)+(0,-.1)$) {$f_{i_r-1}$};
\node (iu2l) at ($(iu2.south)+(0,-.1)$) {$f_{i_r+1}$};
\node (iu3l) at ($(iu3.south)+(0,-.1)$) {$f_{i_r}$};

\node (spoiler1) [right = 1.8 of iu] [draw,circle , fill, red!60!white,opacity=1, label, minimum size = 10 ]{};
\path[-, draw, thick]  (u) -- ($(u)+(\wlength,0)$)  node (spoiler1p) [pos= 0.4, draw,circle , fill, red!60!white,opacity=1, label, minimum size = 10 ]{};
\node at ($(spoiler1p)+(0,-0.5)$) {$i_r$};
\path[-, draw, thick]  (v) -- ($(v)+(\wlength,0)$) node (Dupl1) [pos= 0.5, draw,circle , fill, green!60!black,opacity=1, label, minimum size = 10 ]{};
\node at ($(Dupl1)+(0,-0.5)$) {$j_r$};

\node (ivl) at ($(iv.south)+(0,-.1)$) {$g_{j_r-1}$};
\node (iv3l) at ($(iv3.south)+(0,-.1)$) {$g_{j_r}$};
\node (iv2l) at ($(iv2.south)+(0,-.1)$) {$g_{j_r+1}$};
\node (dotv0) at ($(iv.south west)+(3,-0.4)$) {};
\node (dotv1) at ($(iv.south west)+(0,-0.4)$) {};
\node (dotv2) at ($(iv.south west)+(5.8,-0.4)$) {};

\node (spoiler1) [right =1.5 of iv] [draw,circle , fill, green!60!black,opacity=1, label, minimum size = 10 ]{};

\end{tikzpicture}
		 	 	\end{center}
		 		\item We assume the configuration of the game on
		 		$(u',v')$ to be winnable for Spoiler: he has a winning strategy in
		 		less than $(s-r)$ rounds.
		 	\end{itemize}
			We are going to distinguish two cases.
			Either Duplicator is going to answer on Spoiler's latest move in the game on $(u,v)$
			or Spoiler wins the game. Since we seek a winning strategy for Spoiler, we assume
			that Duplicator successfully answers on $(u,v)$. If this is true, then
			we are going to find an adequate answer for Duplicator in the game on $(u',v')$.
			Since Spoiler has a winning strategy for this latter game, Duplicator
			 eventually loses the game on $(u',v')$ and therefore the game on $(u,v)$.
			 We
		 	remark that the  number of alternations of the
			new Spoiler's winning strategy on $(u,v)$
			is at most 	the one of his strategy on $(u',v')$.
		 	This concludes the proof.

			Nevertheless, it remains to be explained how we construct the position of Spoiler
			on $(u,v)$ and how to deduce from a correct answer for Duplicator on $(u,v)$, a correct answer for
			Duplicator on  	$(u',v')$.

		 	We use the Spoiler's  winning strategy  on $(u',v')$ to construct a new move for Spoiler
		 	on $(u,v)$. Without loss of generality, we assume that this move is on $u'$ and we denote by
		 	$i_{r+1}'$ its position.
		 	We now distinguish four cases that only depend on the value of
		 	$i_{r+1}'$ (see Figure~\ref{fig:dessin_preuve7}).
		 	Indeed, the segment $\{0,\ldots,n-1\}$ is split into four parts that
		 	correspond to the four following cases:
		 	\begin{enumerate}
		 		\item The first case corresponds to segments of the form
		 				$J_{s-r-1}(f_k,f_{k+1})$
		 				for $k\neq i_r$ and $k\neq i_{r-1}$.
		 				It includes almost all the positions of $\{0,\ldots,n\}$ except extremal positions
		 			and a \emph{hole} around positions $i_r'$ and $i_{r-1}'$.
		 		\item The second case corresponds to the \emph{truncated} segment to the left
		 		of the previous pebble on $u'$.
				This is the initial segment of the second constrained game for this position. More precisely
				it is the segment $V(f_{i_r-1},s-r-1)$.
		 		\item The third case corresponds to the allowed positions for the constrained game
		 		around $i_r'$. More precisely, it is the
		 		segment
		 		$I_{(r+1,s)}(f_{i_{r}-1},f_{i_{r}+1}).$
				\item  The last case corresponds to the extremal positions. They are the positions that are
				not handled by the other cases. They are either at the beginning or at the end of the word.
			\end{enumerate}
			The four cases deal with all the positions since the segments of the form
			$J_{s-r-1}(f_k,f_{k+1})$ and the extremal positions form a partition of all the positions.
			Furthermore, by Lemma~\ref{lem:tech_proof}, we have
			\begin{align*}
			J_{s-r-1}(f_{i_r-1},f_{i_r})\cup J_{s-r-1}(f_{i_r},&f_{i_r+1})	= V(f_{i_r-1},s-r-1)\cup
			 I_{(r+1,s)}(f_{i_r-1},f_{i_r+1})\enspace.
			\end{align*}

		 	\begin{figure*}[h]
		 		\centering
		 		\usetikzlibrary{calc,automata,arrows,chains,matrix,positioning,scopes,decorations.pathreplacing}

\begin{tikzpicture}[scale = 0.7,every node/.style={transform shape}]
\def\wlength{20}
\node (up){$u'$};
\node [draw,circle , fill, red!60!white,opacity=1, label, minimum size = 1 ] (old_pebble) [above = 3 of up] {};
\node [right = 0.5of old_pebble] {old pebble position ($i_r'$)};
\node [draw,circle , fill, green!60!black,opacity=1, label, minimum size = 1 ] (new_pebble) [below = 0.5 of old_pebble] {};
\node [right = 0.5 of new_pebble] { new pebble potential positions};

\path[-, draw, thick]  (up) -- ($(up)+(\wlength,0)$) node (iu1) [pos = 0.1] {$\bullet$}  node (iu2) [pos = 0.2] {$\bullet$}
node [draw, rectangle,fill, color = white,] (iu3) [pos = 0.3] {$\cdots$} 
node [draw, rectangle,fill, color = white,] (iu4) [pos = 0.72] {$\cdots$} 
node [draw, rectangle,fill, color = white,] (iu5) [pos = 0.92] {$\cdots$} 
node (ir0) [pos=0.45]{$\bullet$}node (ir) [pos=0.6]{$\bullet$}node (ir1) [pos=0.8]{$\bullet$} 
node (old) [pos=0.66] [draw,circle , fill, red!60!white,opacity=1, label, minimum size = 1 ]{}
node (new1) [pos=0.06] [draw,circle , fill, green!60!black,opacity=1, label, minimum size = 1 ]{}
node (new2) [pos=0.63] [draw,circle , fill, green!60!black,opacity=1, label, minimum size = 1 ]{}
node (new3) [pos=0.86] [draw,circle , fill, green!60!black,opacity=1, label, minimum size = 1 ]{}
node (new4) [pos=0.4] [draw,circle , fill, green!60!black,opacity=1, label, minimum size = 1 ]{}
;
\node at (iu3) {$\cdots$};
\node at (iu4) {$\cdots$};
\node at (iu5) {$\cdots$};

\node (u1l) at ($(iu1.south)+(0,-.1)$) {$i_0$};
\node (u2l) at ($(iu2.south)+(0,-.1)$) {$f_0$};
\node (ir0l) at ($(ir0.south)+(0,-.1)$) {$f_{i_r-1}$};
\node (irl) at ($(ir.south)+(0,-.1)$) {$f_{i_r}$};
\node (ir1l) at ($(ir1.south)+(0,-.1)$) {$f_{i_{r+1}}$};
\draw[ cyan, opacity = 0.1,fill] ($(iu2.north west)+(-0.8,1)$) rectangle  ($(iu2.south east)+(15,-1)$);
\draw[ red, opacity = 0.2,fill] ($(ir0.north west)+(-1,0.6)$) rectangle  ($(ir0.south east)+(1,-0.6)$);
\path[  dashed,draw] ($(iu2.north west)+(-0.8,1)$) --  ($(iu2.north west)+(-0.8,-2)$);
\path[  dashed,draw] ($(iu2.north west)+(-3.7,1)$) --  ($(iu2.north west)+(-3.7,-2)$);
\path[  dashed,draw] ($(ir0.north west)+(-1,1)$) --  ($(ir0.north west)+(-1,-2)$);
\path[  dashed,draw] ($(ir0.north west)+(1.525,1.5)$) --  ($(ir0.north west)+(1.525,-2)$);
\path[  dashed,draw] ($(ir0.north west)+(4.925,1.5)$) --  ($(ir0.north west)+(4.925,-2)$);
\path[  dashed,draw] ($(ir0.north west)+(6.225,1)$) --  ($(ir0.north west)+(6.225,-2)$);
\path[  dashed,draw] ($(ir0.north west)+(8.925,1)$) --  ($(ir0.north west)+(8.925,-2)$);

\draw [decorate, decoration={brace,,amplitude=5}] ($(iu2.north west)+(-0.8,-2)$) -- ($(iu2.north west)+(-3.7,-2)$)  node [ below=3ex,  text height = 1.5ex,   pos =0.5] {left extremal positions}
node [ below=7ex,  text height = 1.5ex,   pos =0.5] {case~\ref{preuve_cas:1}};
\draw [decorate, decoration={brace,mirror,amplitude=5}] ($(ir0.north west)+(-1,-2)$) --  ($(ir0.north west)+(1.525,-2)$)  node [ below=3ex,  text height = 1.5ex,   pos =0.5] {$V(f_{i_r-1},s-r-1)$}
node [ below=7ex,  text height = 1.5ex,   pos =0.5] {case~\ref{preuve_cas:4}};

\draw [decorate, decoration={brace,amplitude=5}] ($(ir0.north west)+(1.525,1.5)$) -- ($(ir0.north west)+(4.925,1.5)$)  node [ above=3ex,  text height = 1.5ex,   pos =0.5] {$I_{(r+1,s)}(f_{i_r-1},f_{i_r+1})$}
node [ above=7ex,  text height = 1.5ex,   pos =0.5] {case~\ref{preuve_cas:2}};

\draw [decorate, decoration={brace,mirror,amplitude=5}]($(ir0.north west)+(6.225,-2)$) -- ($(ir0.north west)+(8.925,-2)$)  node [ below=3ex,  text height = 1.5ex,   pos =0.5] {$J_{s-r-1}(f_{i_{r+1}},f_{i_{r+1}+1})$}
node [ below=7ex,  text height = 1.5ex,   pos =0.5] {case~\ref{preuve_cas:3}};


\end{tikzpicture}
		 		\caption{The four cases to deal with~\label{fig:dessin_preuve7}}
		 		\end{figure*}
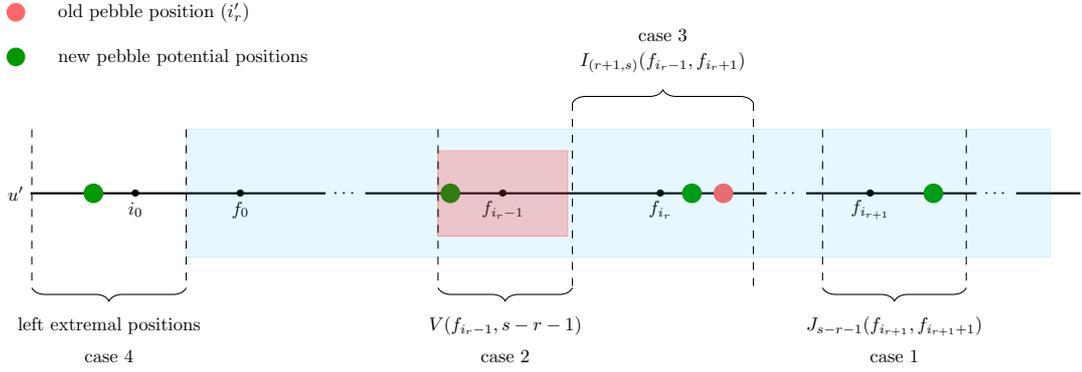
			We now construct the back-and-forth strategy for each of the four cases:
		 	\begin{enumerate}
		 		\item There exists an integer $k$ different from $i_r$ and $i_{r}-1$ such that the position $i_{r+1}'$
		 		belongs to $J_{s-r-1}(f_k,f_{k+1}).$
		 		It is then sufficient for Spoiler to choose $i_{r+1}=k$ on $u$ as its next move for the game
		 		on $(u,v)$. We remark that all the predicates other than  the linear order are
		 		evaluated to  \emph{false} between $i_{r}'$ and $i_{r+1}'$. We assume Duplicator
		 		to be able to answer correctly at a position $j_{r+1}$. We now choose
		 		a position
		 		$j_{r+1}'$ on $v'$ in the set
		 		$J_{s-r-1}(g_{j_{r+1}},g_{j_{r+1}}+1)$
				such that positions $i_{r+1}'$ and $j_{r+1}'$ are locally equivalent for
				the first constrained game. This is possible since positions
				$f_{i_{r+1}}$ and $g_{j_{r+1}}$ are labelled by the same letter and because
				$$(f_{i_{r+1}-1},f_{i_{r+1}},f_{i_{r+1}+1})\sim_s (g_{j_{r+1}-1},g_{j_{r+1}},g_{j_{r+1}+1})\enspace.$$
		 		We remark that all predicates except for the linear order are evaluated as \emph{false} between
		 		$j_r'$ and $j_{r+1}'$. Furthermore, the value of the order predicate
		 		between $i_r'$ and $i_{r+1}'$ is exactly the same as between $i_r$ and $i_{r+1}$
		 		which is also the same as between
		 		$j_r$ and $j_{r+1}$ and between $j_r'$ and $j_{r+1}'$. Since the letters labelling
		 		positions $i_{r+1}$ on $u$ and  $j_{r+1}$ on $v$ are the same, we deduce that
		 		position  $j_{r+1}'$ is correct for Duplicator. Consequently,
		 		the new configuration satisfies the induction hypothesis.
		 		\label{preuve_cas:3}
		 		\item We assume that $i_{r+1}'$ belongs to $V(f_{i_r-1},s-r-1)$.
		 		In this case, we choose $i_{r+1} = i_{r}-1$,
		 		meaning that Spoiler plays on the position just to the left of $i_r$. Since the successor
		 		relation is in the signature, Duplicator is also forced to play at the position immediately
		 		to the left. Here the very same arguments that in case~\ref{preuve_cas:3}
		 		allow us to build a position $j_{r+1}'$ so that
		 		the new configuration satisfies the induction hypothesis hold. The only difference,
		 		is that this time we are using the second constrained game, not the first.\label{preuve_cas:4}
		 		\item  If $i'_{r+1}$ belongs to $I_{(r+1,s)}(f_{i_{r}-1},f_{i_{r}+1}),$
		 		then according to the induction hypothesis, Duplicator
		 		has a position $j'_{r+1}$ in the set
		 		$I_{(r+1,s)}(g_{j_r-1},g_{j_r+1})$ which is locally equivalent to  $i'_{r+1}$.
		 		By choosing this position and by setting
		 		$i_{r+1}= i_{r}$ and $j_{r+1} = j_r$,  we obtain a new configuration that
		 		satisfies the induction hypothesis.
		 		We remark that in this case, the game configuration on $(u,v)$ does not change.\label{preuve_cas:2}

		 		\item The last case is the one which $i_{r+1}'$ does not satisfy
		 		any of the preceding case.
		 		By construction, the positions of the words are split into segments $J_{s-r}(f_k,f_{k+1})$
		 		(resp. $J_{s-r}(g_k,g_{k+1})$) and the extremal positions. Therefore, if the integer
		 		$i_{r+1}$ is not treated by the other cases, then
		 		this position has to be extremal. That is to say
				$$i_{r+1}'<\min J_{s-r-1}(f_0,f_1) = \min J_{s-r-1}(g_0,g_1)$$
				 or
				$$
		 		i_{r+1}'>\max J_{s-r-1}(f_{|u|-1},f_{|u|})
		 		=\max J_{s-r-1}(g_{|v|-1},g_{|v|})\enspace.$$
		 		We choose $j_{r+1}'=i_{r+1}'$ for Duplicator on $v'$, as well as $i_{r+1}=i_r$ and
		 		$j_{r+1}=j_r$. Therefore the game on $(u,v)$ does not evolve and the new configuration
		 		satisfies the induction hypothesis.
				We remark that it is possible for $i_{r+1}'$ to be an extremal position but be handled
				by one of the preceding cases. For instance, if $i_{r}$ belongs to $J_{s-r}(f_0,f_1)$ and if
				$$i_{r+1}\in I_{(r+1,s)}(i_0,f_1)\cap \{0,\ldots, \min J_{s-r-1}(f_0,f_1)\}\enspace,$$
				then Duplicator  follows the first constrained game and it is therefore possible that
				 $i_{r+1}\neq j_{r+1}$. In this particular case, since $i_{r+1}'$ and $j_{r+1}'$ are
				locally equivalent, the configuration still satisfies the induction hypothesis.
		 		\label{preuve_cas:1}
		 	\end{enumerate}
		 	As all the cases are treated, we have proved that as long as Duplicator answers correctly
		 	on $(u,v)$, it is possible for him to answer correctly on $(u',v')$.
		 	Since Spoiler follows a winning strategy on	$(u',v')$, Duplicator will eventually
		 	not be able to answer on $(u,v)$. This concludes the proof.

	\section*{Acknowledgement}

	I am grateful to Olivier Carton for all of his help and support, without which this paper would
	not have been possible. I also thank  Michaël Cadilhac, Amy
	Hadfield, and the anonymous reviewers for their contribution
	in improving a lot the final version of this paper.

\bibliographystyle{plain}
\bibliography{biblio}
\end{document}